\documentclass[12pt]{amsart}

\usepackage{moreverb}
\usepackage{cite}

\usepackage{graphicx}
\usepackage{amsfonts}
\usepackage{amssymb}
\usepackage{latexsym}
\usepackage{amsmath}
\usepackage{color}

\usepackage{float}

\newtheorem{theorem}{Theorem}[section]
\newtheorem{lemma}{Lemma}[section]

\newcommand\BibTeX{{\rmfamily B\kern-.05em \textsc{i\kern-.025em b}\kern-.08em
T\kern-.1667em\lower.7ex\hbox{E}\kern-.125emX}}

\newcommand{\figref}[1]{Figure~\ref{#1}}

    \usepackage[breakable]{tcolorbox}
    \tcbset{nobeforeafter} 
    \usepackage{float}
    \usepackage[T1]{fontenc}
    \usepackage{mathpazo}

    \usepackage{graphicx}
    \makeatletter
    \def\maxwidth{\ifdim\Gin@nat@width>\linewidth\linewidth
    \else\Gin@nat@width\fi}
    \makeatother
    \let\Oldincludegraphics\includegraphics
    \renewcommand{\includegraphics}[1]{\Oldincludegraphics[width=.8\maxwidth]{#1}}
    \usepackage{caption}

    \usepackage{adjustbox} 
    \usepackage{xcolor} 
    \usepackage{enumerate} 
    \usepackage{geometry} 
    \usepackage{amsmath} 
    \usepackage{amssymb} 
    \usepackage{textcomp} 
    \AtBeginDocument{%
        \def\PYZsq{\textquotesingle}
    }
    \usepackage{upquote} 
    \usepackage{eurosym} 
    \usepackage[mathletters]{ucs} 
    \usepackage[utf8x]{inputenc} 
    \usepackage{fancyvrb} 
    \usepackage{grffile} 
    \usepackage{hyperref}
    \usepackage{longtable} 
    \usepackage{booktabs}  
    \usepackage[inline]{enumitem} 
    \usepackage[normalem]{ulem} 
    \usepackage{mathrsfs}

\begin{document}
\title[Average predictive control]{Average predictive control for nonlinear discrete dynamical systems }

\author{D. Dmitrishin}
\address{Department of Applied Mathematics\\Odessa National Polytechnic University\\Odessa 65044, Ukraine}
\email{dmitrishin@opu.ua}

\author{I.E. Iacob}
\address{Department of Mathematical Sciences\\Georgia Southern University\\Statesboro, GA 30460, USA}
\email{ieiacob@GeorgiaSouthern.edu}

\author{A. Stokolos}
\address{Department of Mathematical Sciences\\Georgia Southern University\\Statesboro, GA 30460, USA}
\email{astokolos@GeorgiaSouthern.edu}




\begin{abstract}
We explore the problem of stabilization of unstable periodic orbits in discrete nonlinear dynamical systems. This work proposes the generalization  of predictive control method for resolving the stabilization problem. Our method embodies the development of control method proposed by B.T. Polyak. The control we propose uses a linear (convex) combination of iterated functions. With the proposed method auxiliary, the problem of robust cycle stabilization for various cases of its multipliers localization is solved. An algorithm for finding a given length cycle when its multipliers are known is described as a particular case of our method application. Also, we present numerical simulation results for some well-known mappings and the possibility of further generalization of this method.

\end{abstract}

\maketitle

\section{Introduction}
\label{intro}
Nonlinear dynamical systems are often characterized by extremely unstable movements in the phase space, defined as chaotic movements \cite{1}.
In practice, it is generally desirable to suppress or prevent such chaotic behavior due to its adverse effect on the physical systems normal operation. Due to its theoretical significance and engineering applicability, much attention has been paid to the problem of chaos controlling in various fields and numerous studies \cite{29,3,6n,2}. By chaos control we mean  small external influences on the system or a small change in the system structure in order to transform the system chaotic behavior into a regular (or chaotic, still specific with other properties) one \cite{4}.

It is assumed that the dynamical system includes a chaotic attractor, which contains a countable set of unstable cycles with different periods. If, by using some control effect, a certain cycle is stabilized locally, the system path will remain in its neighborhood, i.e. regular movements will be observed in the system. Hence, one of the ways for chaos controlling refers to the local stabilization of certain orbits from a chaotic attractor.

The problem of stabilizing cycles is closely related to the problem of finding periodic points. The various control schemes \cite{6,30,20} that were proposed for solving these problems can be divided into two large groups: direct and indirect methods. The indirect methods either use the initial mapping $T$ iterations or imply the construction a system of which order is $T$  times greater than the initial system order ($T$ being the desired cycle length). Then one of the methods of finding a fixed point is applied. The most common among fixed point finding techniques is the Newton-Raphson relaxation method and its further modifications \cite{9,7,2,8}. The next step is to select periodic points from the entire set of fixed points. In direct methods, all the points in the cycle are found concurrently, i.e. the whole cycle is stabilized. In this case, the initial system is closed by control, based on the feedback principle \cite{12,13,10,11}. Among such control schemes, the most simple in terms of physical implementation are the linear ones. However, they have significant limitations, as they can only be applied to a narrow domain of the space of parameters that are part of the initial nonlinear system \cite{1n,3n,14}.  To overcome the restricted stabilization condition
Ushio and Yamamoto \cite{5n} introduced a prediction-based feedback control method for discrete chaotic systems with accurate mathematical model. B.T. Polyak \cite{15} (see also \cite{PG}) proposed a direct predictive control method that works well for one-dimensional as well as multidimensional maps. L. Shalby \cite{S} used predictive feedback control method for stabilization of continuous time systems.

Another possible control schemes classification into two groups: methods using the Jacobi matrix and methods not based on this matrix. Naturally, it is assumed that the Jacobi matrix at the cycle points is not properly known, otherwise it would be possible to use the whole powerful apparatus of the linear control theory applied to systems linearized in the cycle neighborhood. The Jacobi matrix is an indispensable attribute of Newton-Raphson-type methods. This matrix is also used in one of the modifications derived from Polyak predictive control method.

One of the main disadvantages of Polyak scheme refers to the need for knowing the Jacobi matrix for the cycle, or at least the need for sufficiently good estimates of the cycle multipliers. The research exposed herein is purposed to improve the Polyak method by replacing it with mixed predicted values. A Jacobi matrix representation for the $T$ cycle at a controlled system is found through the Jacobi matrix of the same cycle in the initial system. Therefore, the correspondence between the cycle multipliers of the open loop system and those of closed loop system is established.

Below, it is assumed that the initial system cycle multipliers are not exactly known, we know only know the range of their localization. Then the solution of cycles robust stabilization problem for various localization of multipliers  is given, and taking into account these general provisions the Polyak method is considered. It is worth noting that in the general case of complex multipliers, we must know precisely enough their localization regions. At the end, the applications of proposed predictive control scheme to stabilize the cycles of some common systems in Physics literature are considered.

\section{Problem statement.}

Considered is a nonlinear discrete system
\begin{equation}\label{1}
x_{n+1}=f\left(x_n\right),\ x_n\in\mathbb R^m,\ n=1,2,\dots,
\end{equation}
where $f(x)$  is a differentiable vector function of corresponding dimension.
It is assumed that the system \eqref{1} has an invariant convex set  $A$, i.e. if   $\xi\in A$, then
$f(\xi)\in A$. We emphasize that we do not assume that the set $A$ is a minimally convex set.
It is also assumed that this system has one or several unstable $T$ cycles  $\left\{\eta_1,\dots,\eta_T\right\}$, where all vectors   are different and belong to the invariant set  $A$, i.e.  $\eta_{j+1}=f\left(\eta_j\right)$,   $j=1,\dots,T-1$,  $\eta_1=f\left(\eta_T\right)$. The considered unstable cycles multipliers are determined as eigenvalues of the product of Jacobi  matrices $\prod\limits_{j=1}^Tf'\left(\eta_{T-j+1}\right)$  of dimensions $m\times m$   at the cycle points. The matrix  $\prod\limits_{j=1}^Tf'\left(\eta_{T-j+1}\right)$ is called the Jacobi  matrix of the cycle $\left\{\eta_1,\dots,\eta_T\right\}.$ Typically, a priori the cycles of the system \eqref{1} are not  known. Consequently, the spectrum $\left\{\mu_1,\dots,\mu_m\right\}$ of the matrix $\prod\limits_{j=1}^Tf'\left(\eta_{T-j+1}\right)$  is unknown as well. The spectrum elements are called cycle multipliers. Below, we assume that some estimates on the localization set $M$ for the cycle multipliers  are known.

Let us consider the control system
\begin{equation}\label{2}
x_{n+1}=F\left(x_n\right),
\end{equation}
where  $F(x)=\sum\limits_{j=1}^N\theta_jf^{((j-1)T+1)}(x)$,   $f^{(1)}(x)=f(x)$,   $f^{(k)}(x)=f\left(f^{(k-1)}(x)\right)$,  $k=2,\dots,T$. The numbers $\theta_1,\dots,\theta_N$  are real. It can be easily verified that  at $\sum\limits_{j=1}^N\theta_j=1$  the system \eqref{2} also includes the cycle   $\left\{\eta_1,\dots,\eta_T\right\}$. We aim to choose such parameter $N$  and coefficients
$\left\{\theta_1,\dots\theta_N\right\}$ so that the system \eqref{2} cycle $\left\{\eta_1,\dots,\eta_T\right\}$  would be locally asymptotically stable. Naturally, when constructing these coefficients, there will be used information on set $M$ of multipliers localization. It is also desirable \cite{16,17} to fulfill an additional condition: the system \eqref{1} invariant convex set   $A$ must be also invariant for the system \eqref{2}. This requirement will be fulfilled, for example, if  $0\le\theta_j\le1$, $j=1,\dots,N$.

Polyak method \cite{15} utilizes the case  $\theta_1=1$,  $\theta_2=\dots=\theta_{N-2}=0$, $\theta_{N-1}=-\theta_N=\varepsilon$. Regarding the set  $M$, it was assumed that
$M={\mathbb D}\cup\left\{\mu^*\right\}$ where ${\mathbb D}=\{z:\,|z|<1\}$ is the central unit circle on the complex plane, and   $\mu^*$
is a known real number. In the case $m=1$  the required coefficient formula has the form
$\displaystyle\varepsilon=
\frac{1\mp\left(\left|\mu^*\right|/\rho\right)^{-\frac1T}}{\left(\mu^*\right)^{N-2}\left(\mu^*-1\right)}$
where $0<\rho<1$.
In this article the control problem is solved for a wider class of multipliers localization set  $M$.

\section{Constructing the Jacobi matrix for a controlled system}

Investigating stability of $T$ cycles of the system \eqref{2}  consists in constructing of Jacobi matrix $\prod\limits_{j=1}^TF'\left(\eta_{T-j+1}\right)$ of that cycle and studying the eigenvalues of this matrix. To derive the Jacobi matrix, we use the ideas from \cite{15}.

Let  $J_j=f'\left(\eta_j\right)$, $j=1,\dots,T$,  then we write the Jacobi matrix of the system \eqref{1} cycle
$\left\{\eta_1,\dots,\eta_T\right\}$
 as $J=J_T\cdot\ldots\cdot J_1$.
We introduce the following auxiliary matrices:
$$
A_1=I,\ A_2=J_1,\ A_3=J_2\cdot J_1,\ \ldots,\ A_{T-1}=J_{T-1}\cdot\ldots\cdot J_1
$$
$(I\,-\, \text{unity matrix of order }\ m\times m)$;
$$
B_1=J_T\cdot\ldots\cdot J_1=J,\ B_2=J_T\cdot\ldots\cdot J_2,\ \ldots, \ B_T=J_T;
$$
then  $B_kA_k=J$,  $k=1,\ldots,T$,  $A_kB_k=\left(J_{k-1}\cdot\ldots\cdot J_1\right)\cdot\left(J_T\cdot\ldots\cdot J_k\right)$ and, consequently,
$\left(A_kB_k\right)^s=A_kJ^{s-1}B_k$,  $s=1,2,\ldots$.

By chain rule:
\begin{equation*}
\left.\left(f^{(s)}(x)\right)'\right|_{x=\eta_i}=
\left.\left(f^{(s-1)}(x)\right)'\right|_{x=\eta_{i+1}}\cdot\left.\left(f(x)\right)'\right|_{x=\eta_i}=
\left.\left(f^{(s-1)}(x)\right)'\right|_{x=\eta_{i+1}}\cdot J_i,
\end{equation*}
we get  $$\left.\left(f^{((j-1)T)}(x)\right)'\right|_{x=\eta_{i}}=A_iJ^{j-2}B_i,\ j=2,\ldots,N$$ and therefore
\begin{equation*}
\left.\left(f^{((j-1)T+1)}(x)\right)'\right|_{x=\eta_{i}}=J_iA_iJ^{j-2}B_i,\ j=2,\ldots,N.
\end{equation*}

Next we find  that $$F'\left(\eta_i\right)=\sum\limits_{j=1}^N\theta_j\left.\left(f^{((j-1)T+1)}(x)\right)'\right|_{x=\eta_i}=
 \theta_1J_i+\sum\limits_{j=2}^N\theta_jJ_iA_iJ^{j-2}B_i.$$
  For the Jacobi matrix of the system \eqref{2} cycle
 $\left\{\eta_1,\ldots,\eta_T\right\}$   we can write:
 \begin{multline*}
 F'\left(\eta_T\right)\cdot\ldots\cdot F'\left(\eta_1\right)=
 J_T\left(\theta_1I+A_T\left(\sum\limits_{j=2}^N\theta_jJ^{j-2}\right)B_T\right)\cdot\\
 J_{T-1}\left(\theta_1I+A_{T-1}\left(\sum\limits_{j=2}^N\theta_jJ^{j-2}\right)B_{T-1}\right)\cdot\ldots\cdot
J_1\left(\theta_1I+A_1\left(\sum\limits_{j=2}^N\theta_jJ^{j-2}\right)B_1\right)
 \end{multline*}
 Taking into account that
\begin{equation*}
J_k\left(\theta_1I+A_k\left(\sum\limits_{j=2}^N\theta_jJ^{j-2}\right)B_k\right)=
\end{equation*}
$$
J_kA_k\left(\theta_1I+\left(\sum\limits_{j=2}^N\theta_jJ^{j-2}\right)B_kA_k\right)A_k^{-1}=
J_kA_k\left(\sum\limits_{j=1}^N\theta_jJ^{j-1}\right)A_k^{-1}
$$
and   $J_kA_k=A_{k+1}$
it follows that
\begin{multline*}
F'\left(\eta_T\right)\cdot\ldots\cdot F'\left(\eta_1\right)=\\
=
J_TA_T\left(\sum\limits_{j=1}^N\theta_jJ^{j-1}\right)A_T^{-1}\cdot J_{T-1}A_{T-1}\left(\sum\limits_{j=1}^N\theta_jJ^{j-1}\right)A_{T-1}^{-1}\cdot\ldots\cdot\\
\cdot J_1A_1\left(\sum\limits_{j=1}^N\theta_jJ^{j-1}\right)A_1^{-1}=
J\left(\sum\limits_{j=1}^N\theta_jJ^{j-1}\right)^T
\end{multline*}

(The reader is gently advised that the superscript $T$ in the formula above and all subsequent formulas denotes power, not transpose.)

For deriving the Jacobian formula above it was assumed that the matrix $J$ was not degenerated.
This limitation can be easily circumvented using a well-known topological technique: considering the matrix $J+\delta I$ instead of the degenerated matrix $J$
 and after all calculations taking the limit as $\delta\rightarrow 0$. Thus, the following result is obtained.

\begin{lemma}\label{lem1}
The Jacobi matrix of the cycle $\left\{\eta_1,\ldots,\eta_T\right\}$ in the system \eqref{2}  can be represented as
\begin{equation}\label{3}
J\left(\sum\limits_{j=1}^N\theta_jJ^{j-1}\right)^T,
\end{equation}
where $J$  is the Jacobi matrix of the  cycle $\left\{\eta_1,\ldots,\eta_T\right\}$ in the system \eqref{1}.
\end{lemma}

We now consider another control system, instead of system \eqref{2}:
\begin{equation}\label{4}
x_{n+1}=f\left(\theta_1x_n+\sum\limits_{j=2}^N\theta_jf^{((j-1)T)}\left(x_n\right)\right).
\end{equation}

When $\sum\limits_{j=1}^N\theta_j=1$ then the system \eqref{4} preserves  the cycle   $\left\{\eta_1,\ldots,\eta_T\right\}.$
In addition, according to formula \eqref{3}, the Jacobi matrix of the system \eqref{4} cycle is expressed in the terms of Jacobi matrix of the system \eqref{1}.
The advantage of the control system \eqref{4} over the system \eqref{2} consists of a fewer calculation of the values for function $f(x)$  (more precisely, the difference is $N-2$).

\section{Main result}

All results presented in this section are formulated for system \eqref{2}, however they hold without change for system \eqref{4}.

\begin{theorem}\label{theo1}
Suppose $f\in C^1$ and that the system \eqref{1} has an unstable $T$ cycle with multipliers
$\left\{\mu_1,\ldots,\mu_m\right\}$. Then this cycle will be a locally asymptotically stable cycle of the system \eqref{2} if
\begin{equation*}
\mu_j\left[r\left(\mu_j\right)\right]^T\in {\mathbb D},\quad j=1,\ldots,m,
\end{equation*}
where 
$r(\mu)=\sum\limits_{j=1}^N\theta_j\mu^{j-1}$.
\end{theorem}
\begin{proof} According to Lemma~\ref{lem1}, the characteristic polynomial for a system of linear approximation in the cycle neighbourhood in
the case of system \eqref{2} can be written as $\varphi(\lambda)=\text{det}\,\left(\lambda I-J[r(J)]^T\right)$.
 By reducing the matrix $J$ to the Jordan form, this characteristic polynomial can be represented as
 $\varphi(\lambda)=\prod\limits_{j=1}^m\left(\lambda I-\mu_j\left[r\left(\mu_j\right)\right]^T\right)$,
 from where the theorem conclusion follows.
\end{proof}

Note that the condition $r(1)=1$  is obligatory.
If, additionally, $\theta_j\in[0,1]$ for  $j=1,\ldots,N$, then $\mu_j\left[r\left(\mu_j\right)\right]^T\in\overline{{\mathbb D}}$
when $\mu_j\in\overline{{\mathbb D}}$, and hence
$\left|\mu_j\left[r\left(\mu_j\right)\right]^T\right|<\left|\mu_j\right|^{1+T}$.
This means that if some multiplier of the system \eqref{1} cycle lies in the unit circle, the corresponding multiplier of the system \eqref{2} will lie closer to zero. Thus for the closed loop system, the stabilization quality is improving.
Various estimates for multipliers allow us to construct control systems that stabilize cycles.

\subsection{Case $M=\left\{\mu_1,\dots,\mu_m\right\}$}

If the multipliers are exactly known, we can choose $N=m+1$  and the coefficients
$\left\{\theta_1,\ldots,\theta_{m+1}\right\}$ from the condition
$r(\mu)=\sum\limits_{j=1}^{m+1}\theta_j\mu^{j-1}=\frac1{\prod\limits_{k=1}^m\left(1-\mu_k\right)}\prod\limits_{k=1}^m\left(\mu-\mu_k\right)$.
Then from Theorem \ref{theo1} we get the following conclusion.

\textit{Conclusion.} Suppose that $f\in C^1$  and the system \eqref{1} has an unstable $T$ cycle with multipliers
$\left\{\mu_1,\ldots,\mu_m\right\}$, and the coefficients
$\theta_1,\dots,\theta_{m+1}$ are found as exposed above.
Then this cycle will be a locally asymptotically stable cycle of system \eqref{2}.
Moreover, if the initial point belongs to the cycle basin of attraction, the convergence to the cycle is superlinear.
	
	The superlinearity of the convergence rate follows from the fact that all
multipliers of system \eqref{2} $\left\{\eta_1,\ldots,\eta_T\right\}$  cycle turn out to be zero.

 Note that the authors are unaware about any other method that allow to stabilize a cycle by knowing the cycle multipliers only. Unfortunately, in a typical situation the multipliers are either unknown. The best we can expect is to localize them approximately. What to do in that case is considered in the next section.

\subsection{Case $M=\{z:\, {\rm Re}\,z\le0\}\cup {\mathbb D}$}

\begin{theorem}\label{theo2}
Suppose $f\in C^1$ and that system \eqref{1} has an unstable $T$-cycle with multipliers
$\left\{\mu_1,\ldots,\mu_m\right\}$  satisfying the conditions:
\begin{equation*}
\left|\mu_j-\widehat{\mu}_j\right|<\delta_j,\ \text{Re}\,\left\{\mu_j\right\}\le0,\ j=1,\ldots,n_1,\ \left|\mu_j\right|<1,\ j=n_1+1,\ldots,m.
\end{equation*}
Let the coefficients  $\theta_j$,   $j=1,\ldots,N$, of the system \eqref{2} be determined from the condition
\begin{equation*}
\sum\limits_{j=1}^{n_1+1}\theta_j\mu^{j-1}=
\frac1{\prod\limits_{k=1}^{n_1}\left(1-\widehat{\mu}_k\right)}\prod\limits_{k=1}^{n_1}\left(\mu-\widehat{\mu}_k\right)
\quad(\text{here\ } N=n_1+1).
\end{equation*}
	Then, for sufficiently small values  $\delta_j$,  $j=1,\ldots,n_1$,  the $T$-cycle will be a locally asymptotically stable cycle of system \eqref{2}.
\end{theorem}

\begin{proof}
 Since $\text{Re}\,\left\{\mu_j\right\}<0$,  $j=1,\ldots,n_1$, then all coefficients $\theta_j>0$,  $j=1,\ldots,n_1$.
 That means that $\left|\mu\left[r\left(\mu\right)\right]^T\right|<|\mu|^{1+T}$ when $|\mu|<1$, i.e. the eigenvalues of the  Jacobi matrix of system \eqref{2}
  cycle corresponding to multipliers $\mu_j$,  $j=n_1+1,\ldots,m$, are smaller than the multipliers absolute values.
  Let  $\delta_j=0$,  $j=1,\ldots,n_1$, then the eigenvalues corresponding to multipliers
  $\mu_j$,  $j=1,\ldots,n_1$, are equal to zero. When  $\delta_j$,  $j=1,\ldots,n_1$, are sufficiently small in magnitude,
  these eigenvalues will lie in the central unit circle, as follows from Rouche theorem.
  Thus, all eigenvalues will less than 1 in absolute value, which means local asymptotic stability.\end{proof}

\subsection{Case $M={\mathbb D}\cup\left\{\mu^*\right\},\ \left|\mu^*\right|>1$\ \cite{15}}

In \cite{15}, the coefficients $\theta_1,\ldots,\theta_N$ were chosen as $\theta_1=1$,  $\theta_2=\ldots=\theta_{N-2}=0$,
$\theta_{N-1}=-\theta_N=\varepsilon$, where
$\displaystyle\varepsilon=\frac{1\mp\left(\left|\mu^*\right|/\rho\right)^{-\frac1T}}{\left(\mu^*\right)^{N-2}\left(\mu^*-1\right)}$,  $0<\rho<1$.
Such a choice ensures that the multipliers belong to the open central unit interval corresponding to $\mu^*$.
However, the condition $\left|\mu\left[r(\mu)\right]^T\right|<1$ with $|\mu|<1$ is not necessarily satisfied. Nevertheless, the value $\varepsilon$
 can be made arbitrarily small by choosing the number $N$ large.
 And then, from Rouche's theorem, it follows that with a sufficiently large $N$
   the other multipliers will remain within the central unit circle.
   This ensures the local asymptotic stability of the system \eqref{2} cycle.

When it is known that $\mu^*<-1$, the control scheme can be simplified, namely:
\begin{equation}\label{5}
x_{n+1}=\theta_1f\left(x_n\right)+\theta_2f^{(T+1)}\left(x_n\right),
\end{equation}
where   $\displaystyle\theta_1=\frac{\left|\mu^*\right|}{1+\left|\mu^*\right|}$,   $\displaystyle\theta_2=\frac1{1+\left|\mu^*\right|}$.

\subsection{Case $M={\mathbb D}\cup\left\{\mu^*,\overline{\mu}^*\right\},\ \left|\mu^*\right|>1$}

The case of general localization of multipliers $\left\{\mu_1,\ldots,\mu_m\right\}$ for the system \eqref{1} cycle was considered in \cite{15} but only for $T=1$.
In that case, the coefficients $\theta_1,\ldots,\theta_N$ were no longer scalars
but matrices and, as before, were chosen as $\theta_1=I$,  $\theta_2=\ldots=\theta_{N-2}=0$,  $\theta_{N-1}=-\theta_N=\varepsilon$, where $I$ is identity matrix, $0$  is zero matrix, $\varepsilon=S\Lambda S^{-1}$,  $\Lambda=\text{diag}\,\left\{\varepsilon_1,\ldots,\varepsilon_m\right\}$,
$\displaystyle\varepsilon_j=\frac{1+{\rm e}^{\imath \varphi}\left(\rho/\left|\mu_j\right|\right)}{\left(\mu_j\right)^{N-2}\left(\mu_j-1\right)}$,
if $\left|\mu_j\right|<1$, and $\varepsilon_j=0$, if $\left|\mu_j\right|>1$, and $0<\rho<1$, $\varphi\in\{0,\pi\}$ if $\mu_j$  as a real number.
The matrix $S$  consists of the eigenvectors of the Jacobi matrix $J$  for equilibrium point.
Thus, to apply the stabilization method, it is necessary to know not only all the multipliers of the equilibrium,
but also the Jacobi matrix itself. That is impossible when the equilibrium   is not known.

Now, let us we apply the scheme \eqref{2}. Let $\mu^*=\rho{\rm e}^{\imath\varphi}$. Then
\begin{equation*}
r(\mu)=\frac{\left(\mu-\mu^*\right)\left(\mu-\overline{\mu}^*\right)}{\left(1-\mu^*\right)\left(1-\overline{\mu}^*\right)}
=\end{equation*}
\begin{equation*}
\frac{\rho^2}{\rho^2-2\rho\cos\varphi+1}+
\frac{-2\rho\cos\varphi}{\rho^2-2\rho\cos\varphi+1}\mu+
\frac1{\rho^2-2\rho\cos\varphi+1}\mu^2.
\end{equation*}

If the complex number $\mu^*$ lies in the left half-plane, then the coefficients of polynomial $r(\mu)$ are positive,
so $\left|\mu\left[r(\mu)\right]^T\right|<|\mu|^{1+T}$ with $|\mu|<1$. Therefore, each multiplier of system \eqref{2}
cycle lying in the central unit circle turns out to be in absolute value less then the multiplier of the corresponding cycle of the system \eqref{1}.
Also, the multipliers corresponding to $\mu^*$ and $\overline{\mu}^*$ change to zero. If the multipliers $\mu^*$ and $\overline{\mu}^*$
  are not exactly known, but they can be well estimated, then for the coefficients $\theta_1$,  $\theta_2$,  $\theta_3$,
  although different from the calculated ones, the values of the polynomial with these coefficients at points $\mu^*$ and $\overline{\mu}^*$
   will not exceed 1 in absolute value, as follows from the Rouche theorem. The desired control system is
\begin{equation*}
x_{n+1}=\theta_1f\left(x_n\right)+\theta_2f^{(T+1)}\left(x_n\right)+\theta_3f^{(2T+1)}\left(x_n\right),
\end{equation*}
where   $\displaystyle\theta_1=\frac{\rho^2}{\rho^2-2\rho\cos\varphi+1}$,
$\displaystyle\theta_2=\frac{-2\rho\cos\varphi}{\rho^2-2\rho\cos\varphi+1}$,
$\displaystyle\theta_3=\frac1{\rho^2-2\rho\cos\varphi+1}$.

\subsection{Case $T=1,\ M=\lfloor-\mu^*,\mu^*\rfloor,\ M=\lfloor-\mu^*,1\rfloor$}

Suppose  $M=\lfloor-\mu^*,\mu^*\rfloor$. From Theorem \ref{theo1} it follows that in order to stabilize the equilibrium, it would be necessary to construct a polynomial
$\mu r(\mu)$, so that $r(1)=1$  and $|\mu r(\mu)|\le1$ for all $|\mu|<\mu^*$.

\begin{theorem}\label{theo3}
Let $f\in C^1$ and the system \eqref{1} has unstable equilibrium with multipliers
$\left\{\mu_1,\ldots,\mu_m\right\}\subset\left[-\mu^*,\mu^*\right]$.
Let the value $N$  be odd and be chosen from the condition
$\displaystyle \csc\frac\pi{2N}>\mu^*$, and the coefficients $\theta_1,\ldots,\theta_N$  from the condition
\begin{equation*}
\mu r(\mu)=\mu\sum\limits_{j=1}^N\theta_j\mu^{j-1}=
(-1)^{\frac{N-1}2}T_N\left(\mu
\sin\frac\pi{2N}\right),
\end{equation*}
where $T_N(x)$  is the first kind Chebyshev polynomial of odd order $N$.
Then this equilibrium   will be a locally asymptotically stable equilibrium   of system \eqref{2} (modulo a finite number of cases when
$\displaystyle\mu_j=\frac{\cos\pi k/N}{\sin\pi/2N}$ for some $k=1,\ldots,N-1$).
\end{theorem}

The proof follows from the properties of the first kind Chebyshev polynomials: $|\mu r(\mu)|\le1$  at $\displaystyle\left|\mu\sin\frac\pi{2N}\right|\le1$,   $r(1)=1$.Note that $\mu^*\to\infty\ (N\to\infty)$ with asymptotics  $\displaystyle\frac2\pi N$.

Now we will consider the case of  $M=\lfloor-\mu^*,1\rfloor$.

\begin{theorem}\label{theo4}
Let $f\in C^1$ and the system \eqref{1} has unstable equilibrium with multipliers
$\left\{\mu_1,\ldots,\mu_m\right\}\subset\lfloor-\mu^*,1\rfloor$.
Let the $N$ value be chosen from the condition $\cot^2\frac\pi{4N}>\mu^*$, and the coefficients $\theta_1,\ldots,\theta_N$  from the conditions
\begin{equation*}
\mu r(\mu)=\mu\sum\limits_{j=1}^N\theta_j\mu^{j-1}=
T_N\left(\mu\left(1-\cos\frac\pi{2N}\right)+\cos\frac\pi{2N}\right),
\end{equation*}
where  $T_N(x)$ is the first kind Chebyshev polynomial of order $N$.
Then this equilibrium will be locally asymptotically stable equilibrium of the system \eqref{2} (modulo a finite number of cases).
\end{theorem}

\begin{proof} Note that $r(0)=T_N(\cos\pi/2N)=0$,  $r(1)=T_N(1)=1$.
  In addition $\left|T_N(x)\right|\le1$, at $|x|\le1$, whence
$|\mu r(\mu)|\le1$ at $\displaystyle\left|\mu\left(1-\cos\frac\pi{2N}\right)+\cos\frac\pi{2N}\right|\le1$.
The last inequality is equivalent to  $\displaystyle-\cot^2\frac\pi{4N}\le\mu\le1$, which proves the theorem.\end{proof}

Note that $\mu^*\to\infty\ (N\to\infty)$ with asymptotics $\displaystyle\frac{16}{\pi^2}N^2$.

\subsection{The general case}

Using the ideas from Theorem \ref{theo1} cases,
we can propose the following $T$-cycle stabilization scheme, for which the coefficients $\theta_j$ are not necessarily constants:

\textit{a)}
find the matrix $f'(x)$,

\textit{b)}
find the vectors  $f^{(s)}(x)$,  $s=1,\ldots,T-1$,

\textit{c)}
find the matrix
$f'\left(f^{(T-1)}(x)\right)\cdot\ldots\cdot f'(f(x))\cdot f'(x)$

\textit{d)}
find the matrix  characteristic polynomial $\sum\limits_{j=1}^{m+1}\theta_j(x)\mu^{j-1}$,

\textit{e)}
normalize the characteristic polynomial
$\displaystyle\frac1{\sum\limits_{j=1}^{m+1}\theta_j(x)}\sum\limits_{j=1}^{m+1}\theta_j(x)\mu^{j-1}$,

\textit{f)}
build the control system
\begin{equation*}
x_{n+1}=F\left(x_n\right),
\end{equation*}
where
$$F(x)=\frac1{\sum\limits_{j=1}^{m+1}\theta_j(x)}\sum\limits_{j=1}^{m+1}\theta_j(x)f^{((j-1)T+1)}(x)$$   or
$$F(x)=f\left(\frac1{\sum\limits_{j=1}^{m+1}\theta_j(x)}\left(\theta_1x+\sum\limits_{j=2}^{m+1}\theta_j(x)f^{((j-1)T)}(x)\right)\right)$$

	Let us consider how this scheme looks like in the case of a linear problem.
Let $f(x)=Ax$ where $A$ is a non-degenerate $m\times m$   matrix.
Then $\eta=0$  is a single fixed point, in the absence of any higher order cycles.
We choose $\theta_j$ from the condition
$\displaystyle\frac1{\text{det}\,(I-A)}\text{det}\,(\mu I-A)=\sum\limits_{j=1}^{m+1}\theta_j\mu^{j-1}$.
Then the control system is
$\displaystyle x_{n+1}=\frac1{\sum\limits_{j=1}^{m+1}\theta_j}\sum\limits_{j=1}^{m+1}\theta_jA^jx_n$.
By the Hamilton-Cayley theorem it follows that this system right-hand side is an identical zero.

In the general case, applying this method to stabilizing chaotic motion tending to mixing,
one can expect that after a certain number of iterations the trajectory falls into the basin of attraction for the stabilized cycle. Then the convergence to the cycle will be superlinear.

	Note that if in all the considered cases $\left|\theta_j\right|$ is  being used instead of $\theta_j$, then
it becomes possible to stabilize the system \eqref{1} cycles with multipliers lying in
$M={\mathbb D}\cup\{\mu:\,\text{Re}\,(\mu)\le0\}$.
Moreover, the convex invariant set of system \eqref{1} will remain such for system \eqref{2}.
In addition, the system \eqref{2} multipliers, corresponding to those multipliers of system \eqref{1} that lie in the unit circle, will become closer to zero.

\section{Examples}

Let us illustrate the effectiveness of the generalized predictive control method for finding periodic orbits with
several well-known examples of scalar and vector chaotic systems \cite{18}. We have experimented with various number of cycles (28, 50, 101, 1001, etc.) using Maple and Python. The results we include here are all for $T=101$ (sections \ref{ssect5.1}--\ref{ssect5.12}), with a precision of 250 decimals. More results, including the Python code to replicate all results, can be found in the Appendix. 
It is essential to note that stabilization took only a few iterations (less than 10) for the majority of the systems.

The scheme applied for the logistic and triangular mappings (examples from sections \ref{ssect5.1}--\ref{ssect5.2}) was a general scheme
\begin{equation*}
\left\{\begin{array}{ll}
x_{n+1}=\frac{\theta\left(x_n\right)}{1+\theta\left(x_n\right)}f\left(x_n\right)+\frac1{1+\theta\left(x_n\right)}f^{(T+1)}\left(x_n\right),\\
\theta\left(x_n\right)=-f'\left(f^{(T-1)}\left(x_n\right)\right)\cdot\ldots\cdot f'\left(f\left(x_n\right)\right)\cdot f'\left(x_n\right).
\end{array}
\right.
\end{equation*}
It was possible to find a large number of cycles for all considered periods $T$;
in general, different initial conditions are producing different cycles.
Numerical calculations show that with sufficiently dense initial values grid, all cycles of a given length can be found.
However, in this case it is necessary to ensure that the point $x_n$  remains within the invariant set, otherwise, as a rule it goes to infinity.
If we use $|\theta(x)|$ instead of  $\theta(x)$, the point $x_n$ will always remain in the invariant set. However, in this case we can find cycles only with multipliers from the set $M={\mathbb D}\cup\{\mu:\,\text{Re}\,(\mu)\le0\}$.


In the two-dimensional case, the scheme used was
\begin{equation}\label{6}
x_{n+1}=f\left(\frac\theta{1+\theta}x_n+\frac1{1+\theta}f^{(T)}\left(x_n\right)\right).
\end{equation} 
The value  $\theta$ should be chosen according to the condition
$$ x\left(\frac\theta{1+\theta}+\frac1{1+\theta}x\right)^T\in {\mathbb D}$$  at  $x=\mu_j^*$,
where $\mu_j^*$  are cycle multipliers $(j=1,2)$, and in general, they are unknown.
In the examples below, one of the multipliers never exceeds one in magnitude, while the second one is negative, greater than one in absolute value.

The Theorem \ref{theo2} guarantees stability conditions even if the parameter $\theta$ satisfies
$$ x\left(\frac\theta{1+\theta}+\frac1{1+\theta}x\right)^T\not=0.$$ It is enough have $\theta$ in the neighborhood of multiplier. In general Theorem \ref{theo2} does not provide the estimate on the parameters, however when one multiplier is in the unit disc and the other is real and negative an elementary trial works quite effectively.

Thus $\theta>0$, and we only have to check the compliance with the condition for the second multiplier
\begin{equation}\label{7}
\left|\mu_2^*\left(\frac\theta{1+\theta}+\frac1{1+\theta}\mu_2^*\right)^T\right|<1.
\end{equation}
Let   $\theta=|\mu_2^*|+\Delta$ and assume that $|\mu_2^*|<2^T$.
If required that $\displaystyle\left|\frac\theta{1+\theta}+\frac1{1+\theta}\mu_2^*\right|<\frac12$,
which is equivalent to  $\displaystyle-\frac13\left(1+\left|\mu_2^*\right|\right)<\Delta<1+\left|\mu_2^*\right|,$ or $\displaystyle\theta\in\left(\frac23\left|\mu_2^*\right|-\frac13,
2\left|\mu_2^*\right|+1\right)$. Now, if $\displaystyle\theta<\frac23\left|\mu_2^*\right|-\frac13$ then $\displaystyle2\theta \le 2\left|\mu_2^*\right|+1$. Therefore, choosing $\theta=2^k$ subsequently for $k=1,2,...$ we are sure that for some $k$ we get  $\displaystyle\theta\in\left(\frac23\left|\mu_2^*\right|-\frac13,
2\left|\mu_2^*\right|+1\right),$ then the condition \eqref{7} will be satisfied.
Thus, the grid for sorting parameter $\theta$ should be chosen rather coarse. This justifies the procedure we used in our examples: running \eqref{6} with small values for $\theta$ and then doubling them until obtaining required cycles. To our surprise, the procedure turns out to be quite efficient. However, this idea has been successfully used recently in adaptive interior-point methods by Lesaja and Potra \cite{lesaja:18:adaptive}.

Therefore, the scheme \eqref{6} allows finding cycles both with small multipliers
(examples from sections \ref{ssect5.3}--\ref{ssect5.8}) and with large ones (examples from sections \ref{ssect5.9}--\ref{ssect5.12}).
In general, the large value of multipliers is not the main obstacle.
More challenging is the problem of small basins of attraction for long cycles.
Therefore, it is convenient to either select a dense grid for initial values or use a sufficiently large number of iterations so that the point $x_n$
would fall into the desired  basin of attraction. 
One can achieve any acceptable accuracy in determining the cyclic point. Our subsequent experiments used precision 250 decimals.

\subsection{Logistic mapping}\label{ssect5.1}

The logistic mapping
\begin{equation}\label{8}	
x_{n+1}=hx_n\left(1-x_n\right),
\end{equation}
is, perhaps, the most popular example. Let us consider the case  $h=3.99$,  $T=101$.
\figref{ris:graph1} illustrates one of its numerous $T=101$-cycles.

\begin{figure}[H]
\centerline{\Oldincludegraphics[scale=0.8]{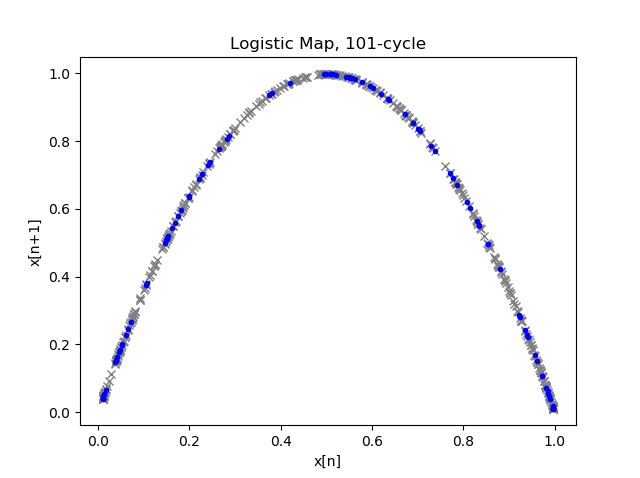}} \caption{A 101-cycle of the logistic mapping \eqref{8}.} \label{ris:graph1}
\end{figure}

\subsection{Triangular mapping}\label{ssect5.2}

Our  next example is the triangular mapping:
\begin{equation}\label{9}
x_{n+1}=h\left(1-\left|2x_n-1\right|\right),\quad h=0.99.
\end{equation}
\figref{ris:graph2} shows a $T=101$-cyclic point.
\begin{figure}[H]
\centerline{\Oldincludegraphics[scale=0.8]{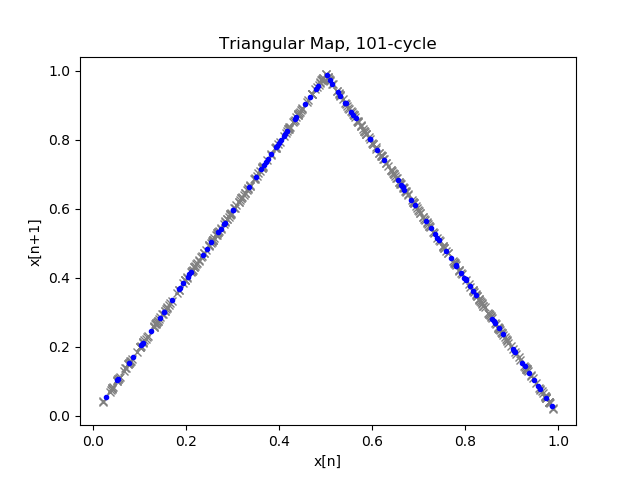}} \caption{A 101-cycle of triangular system \eqref{9}.} \label{ris:graph2}
\end{figure}

\subsection{Burgers mapping}\label{ssect5.3}
For the Burgers mapping:
\begin{equation}\label{10}
x_{n+1}=ax_n-y_n^2,\ y_{n+1}=by_n+x_ny_n,\quad a=0.75,\ b=1.75
\end{equation}
a 101-cyclic point is illustrated in \figref{fig:graph3}.
\begin{figure}[H]
\centerline{\Oldincludegraphics[scale=0.8]{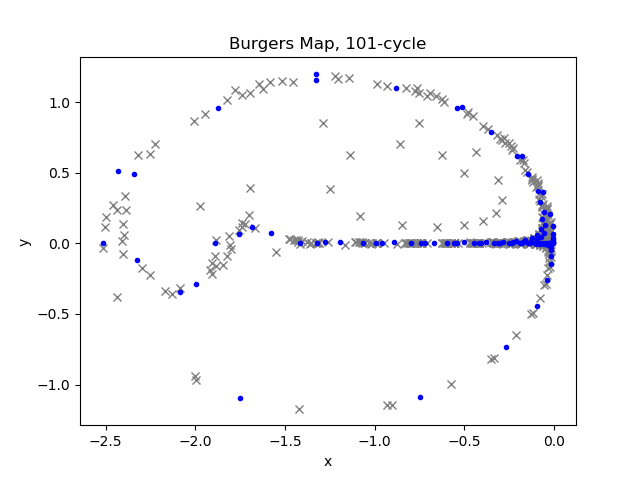}}
 \caption{A 101-cycle of the Burgers system \eqref{10}.}\label{fig:graph3}\end{figure}

\subsection{Tinkerbell mapping}\label{ssect5.4}

The Tinkerbell mapping:
\begin{equation}\label{11}
x_{n+1}=x_n^2-y_n^2+ax_n+by_n,\ y_{n+1}=2x_ny_n+cx_n+dy_n,\; a=0.9,\ b=-0.6,\ c=2.0,\ d=0.5
\end{equation}
has a 101-cyclic point illustrated in \figref{fig:graph4}.
\begin{figure}[H]
\centerline{\Oldincludegraphics[scale=0.8]{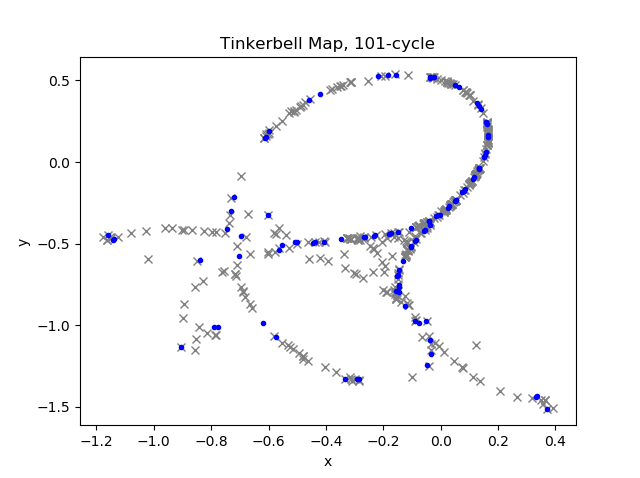}}%
 \caption{A 101-cycle of the Tinkerbell mapping \eqref{11}.}\label{fig:graph4}\end{figure}

\subsection{Gingerbredman mapping}\label{ssect5.5}
The Gingerbredman mapping:
\begin{equation}\label{12}
x_{n+1}=1+\left|x_n\right|-y_n,\ y_{n+1}=x_n
\end{equation}
has a 101-cyclic point represented in \figref{fig:graph5}.
\begin{figure}[H]
\centerline{\Oldincludegraphics[scale=0.8]{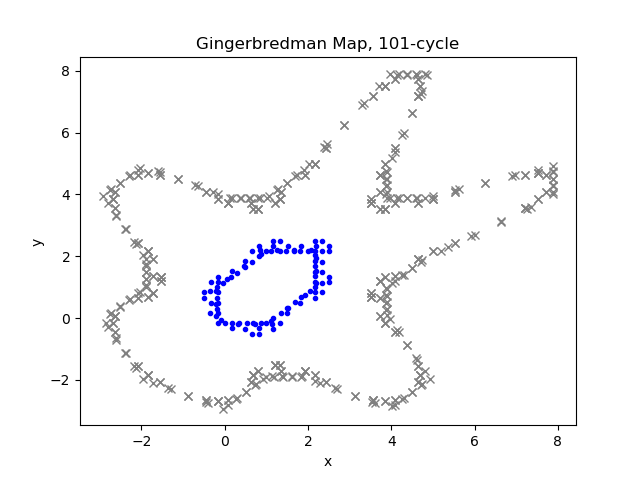}}%
 \caption{A 101-cycle of the Gingerbredman mapping \eqref{12}.}\label{fig:graph5}\end{figure}

\subsection{Prey-predator mapping}\label{ssect5.6}
For the prey-predator mapping:
\begin{equation}\label{13}
x_{n+1}=x_n\exp\left(a\left(1-x_n\right)-by_n\right),\ y_{n+1}=x_n\left(1-\exp\left(-cy_n\right)\right),
\quad a=3,\ b=5,\ c=5
\end{equation}
a corresponding 101-cyclic point is illustrated in \figref{fig:graph6}.

\begin{figure}[H]
\centerline{\Oldincludegraphics[scale=0.8]{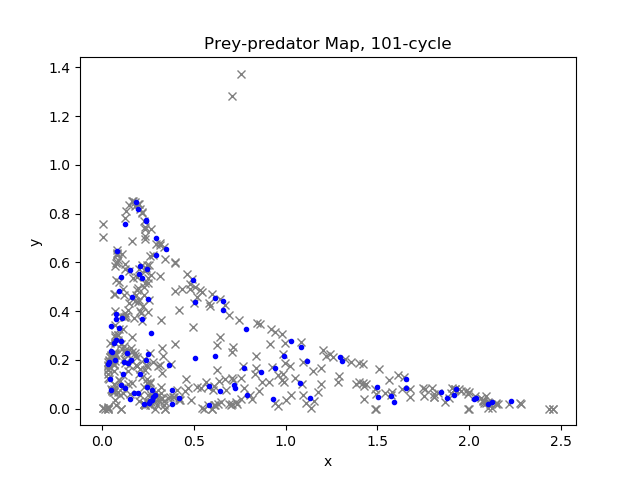}}%
 \caption{A 101-cycle of the prey-predator system \eqref{13}.}\label{fig:graph6}\end{figure}

\subsection{Delayed logistic mapping}\label{ssect5.7}
\figref{fig:graph7} shows a 101-cyclic point of the delayed logistic mapping:
\begin{equation}\label{14}
x_{n+1}=hx_n\left(1-y_n\right),\ y_{n+1}=x_n,\quad h=2.27
\end{equation}

\begin{figure}[H]
\centerline{\Oldincludegraphics[scale=0.8]{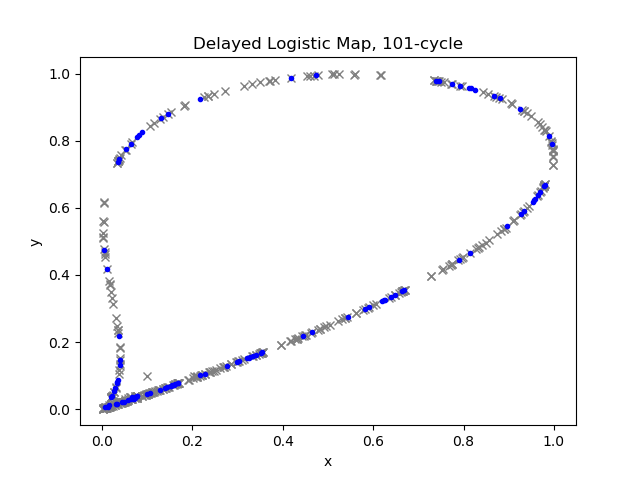}}%
 \caption{A 101-cycle of the delayed logistic mapping system \eqref{14}.}\label{fig:graph7}\end{figure}

\subsection{H\'enon mapping}\label{ssect5.8}
The H\'enon mapping:
\begin{equation}\label{15}
x_{n+1}=1+ax_n^2+y_n,\ y_{n+1}=bx_n,\quad  a=-1.40000001,\ b=0.30000002
\end{equation}
has a 101-cyclic point represented in \figref{fig:graph8}.
\begin{figure}[H]
\centerline{\Oldincludegraphics[scale=0.8]{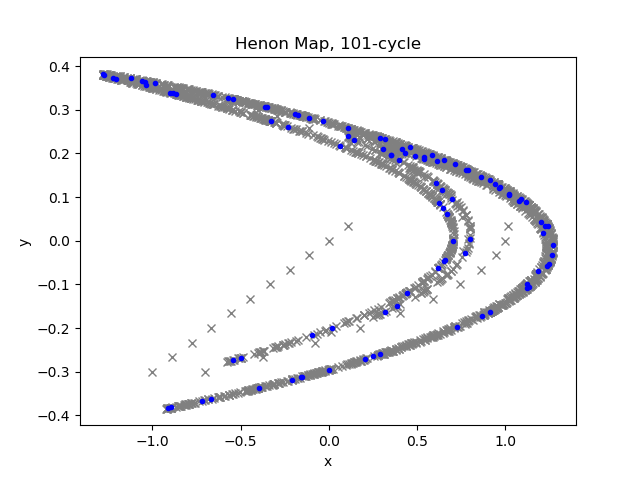}}%
 \caption{A 101-cycle of H\'enon system \eqref{15}.}\label{fig:graph8}\end{figure}

\subsection{Elhadj-Sprott mapping}\label{ssect5.9}
The Elhadj-Sprott mapping:
\begin{equation}\label{16}
x_{n+1}=1+a\sin x_n+by_n,\ y_{n+1}=x_n,\quad a=-4.0,\ b=0.9
\end{equation}
has a 101-cyclic point illustrated in \figref{fig:graph9}.
\begin{figure}[H]
\centerline{\Oldincludegraphics[scale=0.8]{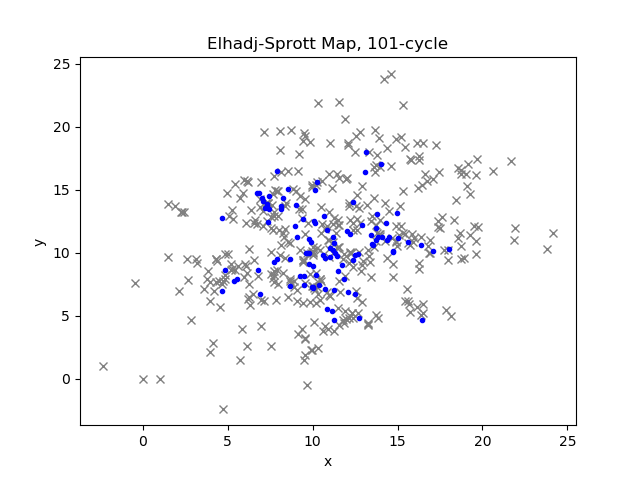}}%
 \caption{A 101-cycle of the Elhadge-Sprott system \eqref{16}.}\label{fig:graph9}\end{figure}

\subsection{Lozi mapping}\label{ssect5.10}
The Lozi mapping:
\begin{equation}\label{17}
x_{n+1}=1+a\left|x_n\right|+by_n,\ y_{n+1}=x_n,\quad a=-1.7,\ b=0.5
\end{equation}
has a 101-cyclic point shown in \figref{fig:graph10}.
\begin{figure}[H]
\centerline{\Oldincludegraphics[scale=0.8]{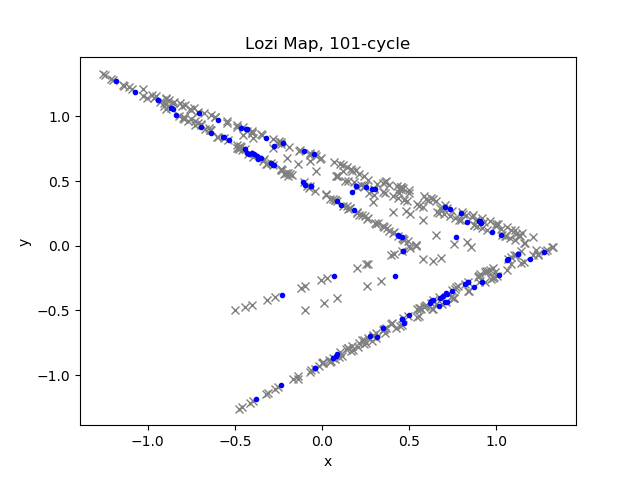}}%
 \caption{A 101-cycle of the Lozi system \eqref{17}.}\label{fig:graph10}\end{figure}

\subsection{Ikeda mapping}\label{ssect5.11}
The Ikeda mapping is given by the equations:
\begin{equation}\label{18}
x_{n+1}=1+u\left(x_n\cos\tau_n-y_n\sin\tau_n\right),\ y_{n+1}=u\left(x_n\sin\tau_n+y_n\cos\tau_n\right),
\end{equation}
where $u=0.9$, $\displaystyle\tau_n=0.4-\frac6{1+x_n^2+y_n^2}$.

The mapping has a 101-cyclic point illustrated in \figref{fig:graph11}.
\begin{figure}[H]
\centerline{\Oldincludegraphics[scale=0.8]{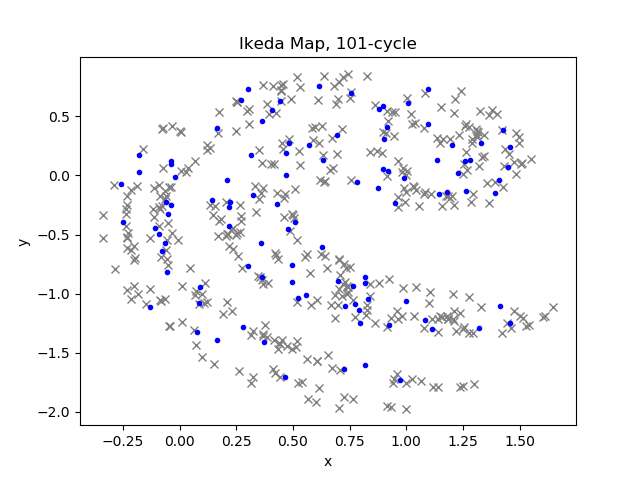}}%
 \caption{A 101-cycle of the Ikeda system \eqref{18}.}\label{fig:graph11}\end{figure}

\subsection{Holmes cubic mapping}\label{ssect5.12}
The Holmes cubic mapping:
\begin{equation}\label{19}
x_{n+1}=y_n,\ y_{n+1}=ax_n+by_n-y_n^3,\quad a=-0.2,\ b=2.77
\end{equation}
has a 101-cyclic point shown in \figref{fig:graph12}.
\begin{figure}[H]
\centerline{\Oldincludegraphics[scale=0.8]{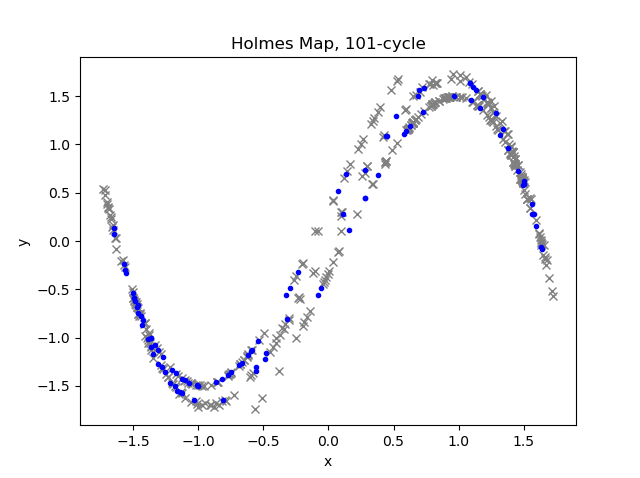}}%
 \caption{A 101-cycle of the Holmes cubic system \eqref{19}.}\label{fig:graph12}\end{figure}

%

\section{Conclusion}

This article deals with the problem of stabilization for nonlinear systems of two categories:
those unstable and those with a priori unknown periodic orbits at discrete time.
A well-known method of stabilizing controls, called the predictive control method,
 first proposed by B.T.Polyak, have been thoroughly investigated in this work.
We have found that this method has several disadvantages: it is necessary to know the cycle exact multiplier or its sufficiently accurate estimate even in the scalar case;
in the vector case, one must know the whole cycle Jacobi matrix;
consequently, the proposed control does not have the required robustness with respect to the system parameters perturbations;
the control gain coefficients have different signs, which can trigger the initial system multiplier's shifting beyond the central unit circle (where it  lies) when applying the control;
therefore, the gain coefficients must be small, and, in order to evaluate them at every instance, we need to know the multiplier's value.

All these shortcomings imply the necessity to modify the predictive control method.
We propose not only to use the first and last iterations of the original mapping, but also all previous ones,
by considering their linear combination.
This linear combination's coefficients are sought as being of a special polynomial, characterized by certain properties. As a result, it was possible for us to extend the predictive control scope.
In addition, if the coefficients are non-negative, then for the initial system cycle multipliers lying in the central unit circle the corresponding multipliers of the control system cycle become closer to zero.
An algorithm is given as a special case of this method application, for finding a cycle of a given length when its multipliers are known.

One of the possible directions for future research is related to investigating
 new control schemes that combine the use of control system previous states and the initial system predicted states,
i.e. the predictive control shall be considered together with the semi-linear control \cite{19} as follows:
\begin{equation}\label{20}
\left\{
\begin{array}{llll}
X_n=\sum\limits_{j=1}^{N_1}a_jx_{n-jT+T}\\
Y_n=\sum\limits_{j=1}^{N_2}b_jx_{n-jT+1}\\
F(x)=\sum\limits_{j=1}^{N_3}\theta_jf^{((j-1)T+1)}(x)\\
x_{n+1}=(1-\gamma)F\left(X_n\right)+\gamma Y_n
\end{array}
\right.
\end{equation}
where  $\sum\limits_{j=1}^{N_1}a_j=1$,  $\sum\limits_{j=1}^{N_2}b_j=1$,  $\sum\limits_{j=1}^{N_3}\theta_j=1$.
Clearly, the $T$-cycles of systems \eqref{1} and \eqref{20} coincide.
The conditions of the system \eqref{20} $T$-cycle local asymptotic stability can be formulated as
\begin{equation*}
\mu_j\left[r\left(\mu_j\right)\right]^T\in\left(\overline{\mathbb C}\setminus\Phi\left(\overline{{\mathbb D}}\right)\right)^*,\quad j=1,\ldots,m,
\end{equation*}
\begin{equation*}
\Phi(z)=(1-\gamma)^T\frac{z(q(z))^T}{(1-\gamma p(z))^T},\quad
q(z)=\sum\limits_{j=1}^{N_1}a_jz^{j-1},\quad
p(z)=\sum\limits_{j=1}^{N_2}b_jz^{j-1},
\end{equation*}
where $\overline{\mathbb C}$  is an extended complex plane, and the asterisk denotes the reciprocal operation:   $\displaystyle(z)^*=\frac1{\overline{z}}$.

	The semilinear control method (when $N_3=1$ in \eqref{20}) has also certain disadvantages \cite{19,20}.
Further studies shall aim to eliminating (reducing)
inherent disadvantages of predictive control and semi-linear control,
synthesizing these approaches together.


\newpage
\section*{Appendix}

\subsection{The Experiments Python Code}
    \definecolor{urlcolor}{rgb}{0,.145,.698}
    \definecolor{linkcolor}{rgb}{.71,0.21,0.01}
    \definecolor{citecolor}{rgb}{.12,.54,.11}

    \definecolor{ansi-black}{HTML}{3E424D}
    \definecolor{ansi-black-intense}{HTML}{282C36}
    \definecolor{ansi-red}{HTML}{E75C58}
    \definecolor{ansi-red-intense}{HTML}{B22B31}
    \definecolor{ansi-green}{HTML}{00A250}
    \definecolor{ansi-green-intense}{HTML}{007427}
    \definecolor{ansi-yellow}{HTML}{DDB62B}
    \definecolor{ansi-yellow-intense}{HTML}{B27D12}
    \definecolor{ansi-blue}{HTML}{208FFB}
    \definecolor{ansi-blue-intense}{HTML}{0065CA}
    \definecolor{ansi-magenta}{HTML}{D160C4}
    \definecolor{ansi-magenta-intense}{HTML}{A03196}
    \definecolor{ansi-cyan}{HTML}{60C6C8}
    \definecolor{ansi-cyan-intense}{HTML}{258F8F}
    \definecolor{ansi-white}{HTML}{C5C1B4}
    \definecolor{ansi-white-intense}{HTML}{A1A6B2}
    \definecolor{ansi-default-inverse-fg}{HTML}{FFFFFF}
    \definecolor{ansi-default-inverse-bg}{HTML}{000000}

    \providecommand{\tightlist}{%
      \setlength{\itemsep}{0pt}\setlength{\parskip}{0pt}}
    \DefineVerbatimEnvironment{Highlighting}{Verbatim}{commandchars=\\\{\}}
    \newenvironment{Shaded}{}{}
    \newcommand{\KeywordTok}[1]{\textcolor[rgb]{0.00,0.44,0.13}{\textbf{{#1}}}}
    \newcommand{\DataTypeTok}[1]{\textcolor[rgb]{0.56,0.13,0.00}{{#1}}}
    \newcommand{\DecValTok}[1]{\textcolor[rgb]{0.25,0.63,0.44}{{#1}}}
    \newcommand{\BaseNTok}[1]{\textcolor[rgb]{0.25,0.63,0.44}{{#1}}}
    \newcommand{\FloatTok}[1]{\textcolor[rgb]{0.25,0.63,0.44}{{#1}}}
    \newcommand{\CharTok}[1]{\textcolor[rgb]{0.25,0.44,0.63}{{#1}}}
    \newcommand{\StringTok}[1]{\textcolor[rgb]{0.25,0.44,0.63}{{#1}}}
    \newcommand{\CommentTok}[1]{\textcolor[rgb]{0.38,0.63,0.69}{\textit{{#1}}}}
    \newcommand{\OtherTok}[1]{\textcolor[rgb]{0.00,0.44,0.13}{{#1}}}
    \newcommand{\AlertTok}[1]{\textcolor[rgb]{1.00,0.00,0.00}{\textbf{{#1}}}}
    \newcommand{\FunctionTok}[1]{\textcolor[rgb]{0.02,0.16,0.49}{{#1}}}
    \newcommand{\RegionMarkerTok}[1]{{#1}}
    \newcommand{\ErrorTok}[1]{\textcolor[rgb]{1.00,0.00,0.00}{\textbf{{#1}}}}
    \newcommand{\NormalTok}[1]{{#1}}

    \newcommand{\ConstantTok}[1]{\textcolor[rgb]{0.53,0.00,0.00}{{#1}}}
    \newcommand{\SpecialCharTok}[1]{\textcolor[rgb]{0.25,0.44,0.63}{{#1}}}
    \newcommand{\VerbatimStringTok}[1]{\textcolor[rgb]{0.25,0.44,0.63}{{#1}}}
    \newcommand{\SpecialStringTok}[1]{\textcolor[rgb]{0.73,0.40,0.53}{{#1}}}
    \newcommand{\ImportTok}[1]{{#1}}
    \newcommand{\DocumentationTok}[1]{\textcolor[rgb]{0.73,0.13,0.13}{\textit{{#1}}}}
    \newcommand{\AnnotationTok}[1]{\textcolor[rgb]{0.38,0.63,0.69}{\textbf{\textit{{#1}}}}}
    \newcommand{\CommentVarTok}[1]{\textcolor[rgb]{0.38,0.63,0.69}{\textbf{\textit{{#1}}}}}
    \newcommand{\VariableTok}[1]{\textcolor[rgb]{0.10,0.09,0.49}{{#1}}}
    \newcommand{\ControlFlowTok}[1]{\textcolor[rgb]{0.00,0.44,0.13}{\textbf{{#1}}}}
    \newcommand{\OperatorTok}[1]{\textcolor[rgb]{0.40,0.40,0.40}{{#1}}}
    \newcommand{\BuiltInTok}[1]{{#1}}
    \newcommand{\ExtensionTok}[1]{{#1}}
    \newcommand{\PreprocessorTok}[1]{\textcolor[rgb]{0.74,0.48,0.00}{{#1}}}
    \newcommand{\AttributeTok}[1]{\textcolor[rgb]{0.49,0.56,0.16}{{#1}}}
    \newcommand{\InformationTok}[1]{\textcolor[rgb]{0.38,0.63,0.69}{\textbf{\textit{{#1}}}}}
    \newcommand{\WarningTok}[1]{\textcolor[rgb]{0.38,0.63,0.69}{\textbf{\textit{{#1}}}}}

    \def\br{\hspace*{\fill} \\* }
    \def\gt{>}
    \def\lt{<}
    \let\Oldtex\TeX
    \let\Oldlatex\LaTeX
    \renewcommand{\TeX}{\textrm{\Oldtex}}
    \renewcommand{\LaTeX}{\textrm{\Oldlatex}}

\makeatletter
\def\PY@reset{\let\PY@it=\relax \let\PY@bf=\relax%
    \let\PY@ul=\relax \let\PY@tc=\relax%
    \let\PY@bc=\relax \let\PY@ff=\relax}
\def\PY@tok#1{\csname PY@tok@#1\endcsname}
\def\PY@toks#1+{\ifx\relax#1\empty\else%
    \PY@tok{#1}\expandafter\PY@toks\fi}
\def\PY@do#1{\PY@bc{\PY@tc{\PY@ul{%
    \PY@it{\PY@bf{\PY@ff{#1}}}}}}}
\def\PY#1#2{\PY@reset\PY@toks#1+\relax+\PY@do{#2}}

\expandafter\def\csname PY@tok@w\endcsname{\def\PY@tc##1{\textcolor[rgb]{0.73,0.73,0.73}{##1}}}
\expandafter\def\csname PY@tok@c\endcsname{\let\PY@it=\textit\def\PY@tc##1{\textcolor[rgb]{0.25,0.50,0.50}{##1}}}
\expandafter\def\csname PY@tok@cp\endcsname{\def\PY@tc##1{\textcolor[rgb]{0.74,0.48,0.00}{##1}}}
\expandafter\def\csname PY@tok@k\endcsname{\let\PY@bf=\textbf\def\PY@tc##1{\textcolor[rgb]{0.00,0.50,0.00}{##1}}}
\expandafter\def\csname PY@tok@kp\endcsname{\def\PY@tc##1{\textcolor[rgb]{0.00,0.50,0.00}{##1}}}
\expandafter\def\csname PY@tok@kt\endcsname{\def\PY@tc##1{\textcolor[rgb]{0.69,0.00,0.25}{##1}}}
\expandafter\def\csname PY@tok@o\endcsname{\def\PY@tc##1{\textcolor[rgb]{0.40,0.40,0.40}{##1}}}
\expandafter\def\csname PY@tok@ow\endcsname{\let\PY@bf=\textbf\def\PY@tc##1{\textcolor[rgb]{0.67,0.13,1.00}{##1}}}
\expandafter\def\csname PY@tok@nb\endcsname{\def\PY@tc##1{\textcolor[rgb]{0.00,0.50,0.00}{##1}}}
\expandafter\def\csname PY@tok@nf\endcsname{\def\PY@tc##1{\textcolor[rgb]{0.00,0.00,1.00}{##1}}}
\expandafter\def\csname PY@tok@nc\endcsname{\let\PY@bf=\textbf\def\PY@tc##1{\textcolor[rgb]{0.00,0.00,1.00}{##1}}}
\expandafter\def\csname PY@tok@nn\endcsname{\let\PY@bf=\textbf\def\PY@tc##1{\textcolor[rgb]{0.00,0.00,1.00}{##1}}}
\expandafter\def\csname PY@tok@ne\endcsname{\let\PY@bf=\textbf\def\PY@tc##1{\textcolor[rgb]{0.82,0.25,0.23}{##1}}}
\expandafter\def\csname PY@tok@nv\endcsname{\def\PY@tc##1{\textcolor[rgb]{0.10,0.09,0.49}{##1}}}
\expandafter\def\csname PY@tok@no\endcsname{\def\PY@tc##1{\textcolor[rgb]{0.53,0.00,0.00}{##1}}}
\expandafter\def\csname PY@tok@nl\endcsname{\def\PY@tc##1{\textcolor[rgb]{0.63,0.63,0.00}{##1}}}
\expandafter\def\csname PY@tok@ni\endcsname{\let\PY@bf=\textbf\def\PY@tc##1{\textcolor[rgb]{0.60,0.60,0.60}{##1}}}
\expandafter\def\csname PY@tok@na\endcsname{\def\PY@tc##1{\textcolor[rgb]{0.49,0.56,0.16}{##1}}}
\expandafter\def\csname PY@tok@nt\endcsname{\let\PY@bf=\textbf\def\PY@tc##1{\textcolor[rgb]{0.00,0.50,0.00}{##1}}}
\expandafter\def\csname PY@tok@nd\endcsname{\def\PY@tc##1{\textcolor[rgb]{0.67,0.13,1.00}{##1}}}
\expandafter\def\csname PY@tok@s\endcsname{\def\PY@tc##1{\textcolor[rgb]{0.73,0.13,0.13}{##1}}}
\expandafter\def\csname PY@tok@sd\endcsname{\let\PY@it=\textit\def\PY@tc##1{\textcolor[rgb]{0.73,0.13,0.13}{##1}}}
\expandafter\def\csname PY@tok@si\endcsname{\let\PY@bf=\textbf\def\PY@tc##1{\textcolor[rgb]{0.73,0.40,0.53}{##1}}}
\expandafter\def\csname PY@tok@se\endcsname{\let\PY@bf=\textbf\def\PY@tc##1{\textcolor[rgb]{0.73,0.40,0.13}{##1}}}
\expandafter\def\csname PY@tok@sr\endcsname{\def\PY@tc##1{\textcolor[rgb]{0.73,0.40,0.53}{##1}}}
\expandafter\def\csname PY@tok@ss\endcsname{\def\PY@tc##1{\textcolor[rgb]{0.10,0.09,0.49}{##1}}}
\expandafter\def\csname PY@tok@sx\endcsname{\def\PY@tc##1{\textcolor[rgb]{0.00,0.50,0.00}{##1}}}
\expandafter\def\csname PY@tok@m\endcsname{\def\PY@tc##1{\textcolor[rgb]{0.40,0.40,0.40}{##1}}}
\expandafter\def\csname PY@tok@gh\endcsname{\let\PY@bf=\textbf\def\PY@tc##1{\textcolor[rgb]{0.00,0.00,0.50}{##1}}}
\expandafter\def\csname PY@tok@gu\endcsname{\let\PY@bf=\textbf\def\PY@tc##1{\textcolor[rgb]{0.50,0.00,0.50}{##1}}}
\expandafter\def\csname PY@tok@gd\endcsname{\def\PY@tc##1{\textcolor[rgb]{0.63,0.00,0.00}{##1}}}
\expandafter\def\csname PY@tok@gi\endcsname{\def\PY@tc##1{\textcolor[rgb]{0.00,0.63,0.00}{##1}}}
\expandafter\def\csname PY@tok@gr\endcsname{\def\PY@tc##1{\textcolor[rgb]{1.00,0.00,0.00}{##1}}}
\expandafter\def\csname PY@tok@ge\endcsname{\let\PY@it=\textit}
\expandafter\def\csname PY@tok@gs\endcsname{\let\PY@bf=\textbf}
\expandafter\def\csname PY@tok@gp\endcsname{\let\PY@bf=\textbf\def\PY@tc##1{\textcolor[rgb]{0.00,0.00,0.50}{##1}}}
\expandafter\def\csname PY@tok@go\endcsname{\def\PY@tc##1{\textcolor[rgb]{0.53,0.53,0.53}{##1}}}
\expandafter\def\csname PY@tok@gt\endcsname{\def\PY@tc##1{\textcolor[rgb]{0.00,0.27,0.87}{##1}}}
\expandafter\def\csname PY@tok@err\endcsname{\def\PY@bc##1{\setlength{\fboxsep}{0pt}\fcolorbox[rgb]{1.00,0.00,0.00}{1,1,1}{\strut ##1}}}
\expandafter\def\csname PY@tok@kc\endcsname{\let\PY@bf=\textbf\def\PY@tc##1{\textcolor[rgb]{0.00,0.50,0.00}{##1}}}
\expandafter\def\csname PY@tok@kd\endcsname{\let\PY@bf=\textbf\def\PY@tc##1{\textcolor[rgb]{0.00,0.50,0.00}{##1}}}
\expandafter\def\csname PY@tok@kn\endcsname{\let\PY@bf=\textbf\def\PY@tc##1{\textcolor[rgb]{0.00,0.50,0.00}{##1}}}
\expandafter\def\csname PY@tok@kr\endcsname{\let\PY@bf=\textbf\def\PY@tc##1{\textcolor[rgb]{0.00,0.50,0.00}{##1}}}
\expandafter\def\csname PY@tok@bp\endcsname{\def\PY@tc##1{\textcolor[rgb]{0.00,0.50,0.00}{##1}}}
\expandafter\def\csname PY@tok@fm\endcsname{\def\PY@tc##1{\textcolor[rgb]{0.00,0.00,1.00}{##1}}}
\expandafter\def\csname PY@tok@vc\endcsname{\def\PY@tc##1{\textcolor[rgb]{0.10,0.09,0.49}{##1}}}
\expandafter\def\csname PY@tok@vg\endcsname{\def\PY@tc##1{\textcolor[rgb]{0.10,0.09,0.49}{##1}}}
\expandafter\def\csname PY@tok@vi\endcsname{\def\PY@tc##1{\textcolor[rgb]{0.10,0.09,0.49}{##1}}}
\expandafter\def\csname PY@tok@vm\endcsname{\def\PY@tc##1{\textcolor[rgb]{0.10,0.09,0.49}{##1}}}
\expandafter\def\csname PY@tok@sa\endcsname{\def\PY@tc##1{\textcolor[rgb]{0.73,0.13,0.13}{##1}}}
\expandafter\def\csname PY@tok@sb\endcsname{\def\PY@tc##1{\textcolor[rgb]{0.73,0.13,0.13}{##1}}}
\expandafter\def\csname PY@tok@sc\endcsname{\def\PY@tc##1{\textcolor[rgb]{0.73,0.13,0.13}{##1}}}
\expandafter\def\csname PY@tok@dl\endcsname{\def\PY@tc##1{\textcolor[rgb]{0.73,0.13,0.13}{##1}}}
\expandafter\def\csname PY@tok@s2\endcsname{\def\PY@tc##1{\textcolor[rgb]{0.73,0.13,0.13}{##1}}}
\expandafter\def\csname PY@tok@sh\endcsname{\def\PY@tc##1{\textcolor[rgb]{0.73,0.13,0.13}{##1}}}
\expandafter\def\csname PY@tok@s1\endcsname{\def\PY@tc##1{\textcolor[rgb]{0.73,0.13,0.13}{##1}}}
\expandafter\def\csname PY@tok@mb\endcsname{\def\PY@tc##1{\textcolor[rgb]{0.40,0.40,0.40}{##1}}}
\expandafter\def\csname PY@tok@mf\endcsname{\def\PY@tc##1{\textcolor[rgb]{0.40,0.40,0.40}{##1}}}
\expandafter\def\csname PY@tok@mh\endcsname{\def\PY@tc##1{\textcolor[rgb]{0.40,0.40,0.40}{##1}}}
\expandafter\def\csname PY@tok@mi\endcsname{\def\PY@tc##1{\textcolor[rgb]{0.40,0.40,0.40}{##1}}}
\expandafter\def\csname PY@tok@il\endcsname{\def\PY@tc##1{\textcolor[rgb]{0.40,0.40,0.40}{##1}}}
\expandafter\def\csname PY@tok@mo\endcsname{\def\PY@tc##1{\textcolor[rgb]{0.40,0.40,0.40}{##1}}}
\expandafter\def\csname PY@tok@ch\endcsname{\let\PY@it=\textit\def\PY@tc##1{\textcolor[rgb]{0.25,0.50,0.50}{##1}}}
\expandafter\def\csname PY@tok@cm\endcsname{\let\PY@it=\textit\def\PY@tc##1{\textcolor[rgb]{0.25,0.50,0.50}{##1}}}
\expandafter\def\csname PY@tok@cpf\endcsname{\let\PY@it=\textit\def\PY@tc##1{\textcolor[rgb]{0.25,0.50,0.50}{##1}}}
\expandafter\def\csname PY@tok@c1\endcsname{\let\PY@it=\textit\def\PY@tc##1{\textcolor[rgb]{0.25,0.50,0.50}{##1}}}
\expandafter\def\csname PY@tok@cs\endcsname{\let\PY@it=\textit\def\PY@tc##1{\textcolor[rgb]{0.25,0.50,0.50}{##1}}}

\def\PYZbs{\char`\\}
\def\PYZus{\char`\_}
\def\PYZob{\char`\{}
\def\PYZcb{\char`\}}
\def\PYZca{\char`\^}
\def\PYZam{\char`\&}
\def\PYZlt{\char`\<}
\def\PYZgt{\char`\>}
\def\PYZsh{\char`\#}
\def\PYZpc{\char`\%}
\def\PYZdl{\char`\$}
\def\PYZhy{\char`\-}
\def\PYZsq{\char`\'}
\def\PYZdq{\char`\"}
\def\PYZti{\char`\~}
\def\PYZat{@}
\def\PYZlb{[}
\def\PYZrb{]}
\makeatother

    \makeatletter
        \newbox\Wrappedcontinuationbox
        \newbox\Wrappedvisiblespacebox
        \newcommand*\Wrappedvisiblespace {\textcolor{red}{\textvisiblespace}}
        \newcommand*\Wrappedcontinuationsymbol {\textcolor{red}{\llap{\tiny$\m@th\hookrightarrow$}}}
        \newcommand*\Wrappedcontinuationindent {3ex }
        \newcommand*\Wrappedafterbreak {\kern\Wrappedcontinuationindent\copy\Wrappedcontinuationbox}
        \newcommand*\Wrappedbreaksatspecials {%
            \def\PYGZus{\discretionary{\char`\_}{\Wrappedafterbreak}{\char`\_}}%
            \def\PYGZob{\discretionary{}{\Wrappedafterbreak\char`\{}{\char`\{}}%
            \def\PYGZcb{\discretionary{\char`\}}{\Wrappedafterbreak}{\char`\}}}%
            \def\PYGZca{\discretionary{\char`\^}{\Wrappedafterbreak}{\char`\^}}%
            \def\PYGZam{\discretionary{\char`\&}{\Wrappedafterbreak}{\char`\&}}%
            \def\PYGZlt{\discretionary{}{\Wrappedafterbreak\char`\<}{\char`\<}}%
            \def\PYGZgt{\discretionary{\char`\>}{\Wrappedafterbreak}{\char`\>}}%
            \def\PYGZsh{\discretionary{}{\Wrappedafterbreak\char`\#}{\char`\#}}%
            \def\PYGZpc{\discretionary{}{\Wrappedafterbreak\char`\%}{\char`\%}}%
            \def\PYGZdl{\discretionary{}{\Wrappedafterbreak\char`\$}{\char`\$}}%
            \def\PYGZhy{\discretionary{\char`\-}{\Wrappedafterbreak}{\char`\-}}%
            \def\PYGZsq{\discretionary{}{\Wrappedafterbreak\textquotesingle}{\textquotesingle}}%
            \def\PYGZdq{\discretionary{}{\Wrappedafterbreak\char`\"}{\char`\"}}%
            \def\PYGZti{\discretionary{\char`\~}{\Wrappedafterbreak}{\char`\~}}%
        }
        \newcommand*\Wrappedbreaksatpunct {%
            \lccode`\~`\.\lowercase{\def~}{\discretionary{\hbox{\char`\.}}{\Wrappedafterbreak}{\hbox{\char`\.}}}%
            \lccode`\~`\,\lowercase{\def~}{\discretionary{\hbox{\char`\,}}{\Wrappedafterbreak}{\hbox{\char`\,}}}%
            \lccode`\~`\;\lowercase{\def~}{\discretionary{\hbox{\char`\;}}{\Wrappedafterbreak}{\hbox{\char`\;}}}%
            \lccode`\~`\:\lowercase{\def~}{\discretionary{\hbox{\char`\:}}{\Wrappedafterbreak}{\hbox{\char`\:}}}%
            \lccode`\~`\?\lowercase{\def~}{\discretionary{\hbox{\char`\?}}{\Wrappedafterbreak}{\hbox{\char`\?}}}%
            \lccode`\~`\!\lowercase{\def~}{\discretionary{\hbox{\char`\!}}{\Wrappedafterbreak}{\hbox{\char`\!}}}%
            \lccode`\~`\/\lowercase{\def~}{\discretionary{\hbox{\char`\/}}{\Wrappedafterbreak}{\hbox{\char`\/}}}%
            \catcode`\.\active
            \catcode`\,\active
            \catcode`\;\active
            \catcode`\:\active
            \catcode`\?\active
            \catcode`\!\active
            \catcode`\/\active
            \lccode`\~`\~ 	
        }
    \makeatother

    \let\OriginalVerbatim=\Verbatim
    \makeatletter
    \renewcommand{\Verbatim}[1][1]{%
        \sbox\Wrappedcontinuationbox {\Wrappedcontinuationsymbol}%
        \sbox\Wrappedvisiblespacebox {\FV@SetupFont\Wrappedvisiblespace}%
        \def\FancyVerbFormatLine ##1{\hsize\linewidth
            \vtop{\raggedright\hyphenpenalty\z@\exhyphenpenalty\z@
                \doublehyphendemerits\z@\finalhyphendemerits\z@
                \strut ##1\strut}%
        }%
        \def\FV@Space {%
            \nobreak\hskip\z@ plus\fontdimen3\font minus\fontdimen4\font
            \discretionary{\copy\Wrappedvisiblespacebox}{\Wrappedafterbreak}
            {\kern\fontdimen2\font}%
        }%

        \Wrappedbreaksatspecials
        \OriginalVerbatim[#1,codes*=\Wrappedbreaksatpunct]%
    }
    \makeatother

    \definecolor{incolor}{HTML}{303F9F}
    \definecolor{outcolor}{HTML}{D84315}
    \definecolor{cellborder}{HTML}{CFCFCF}
    \definecolor{cellbackground}{HTML}{F7F7F7}

    \newcommand{\prompt}[4]{
        \llap{{\color{#2}[#3]: #4}}\vspace{-1.25em}
    }

    \sloppy
    \hypersetup{
      breaklinks=true,  
      colorlinks=true,
      urlcolor=urlcolor,
      linkcolor=linkcolor,
      citecolor=citecolor,
      }




    \begin{tcolorbox}[breakable, size=fbox, boxrule=1pt, pad at break*=1mm,colback=cellbackground, colframe=cellborder]
\prompt{In}{incolor}{3}{\hspace{4pt}}


    \begin{center}
    \adjustimage{max size={0.9\linewidth}{0.9\paperheight}}{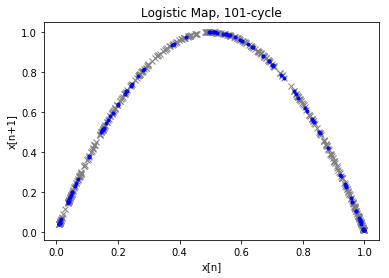}
    \end{center}
    { \hspace*{\fill} \\}

    \begin{center}
    \adjustimage{max size={0.9\linewidth}{0.9\paperheight}}{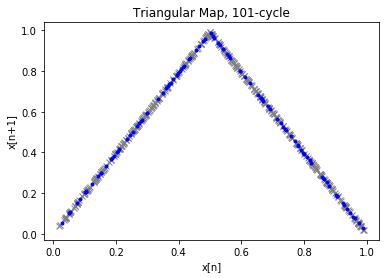}
    \end{center}
    { \hspace*{\fill} \\}

    \begin{center}
    \adjustimage{max size={0.9\linewidth}{0.9\paperheight}}{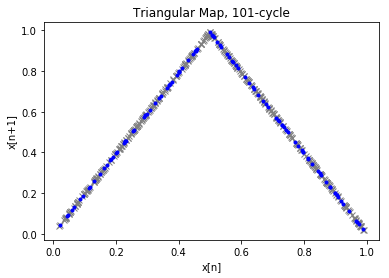}
    \end{center}
    { \hspace*{\fill} \\}

    \begin{center}
    \adjustimage{max size={0.9\linewidth}{0.9\paperheight}}{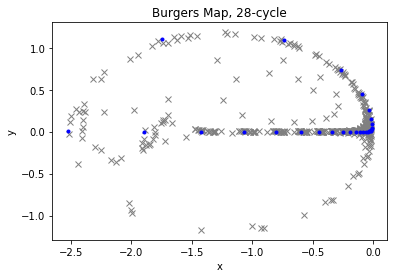}
    \end{center}
    { \hspace*{\fill} \\}

    \begin{center}
    \adjustimage{max size={0.9\linewidth}{0.9\paperheight}}{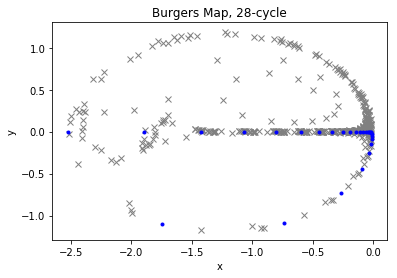}
    \end{center}
    { \hspace*{\fill} \\}

    \begin{center}
    \adjustimage{max size={0.9\linewidth}{0.9\paperheight}}{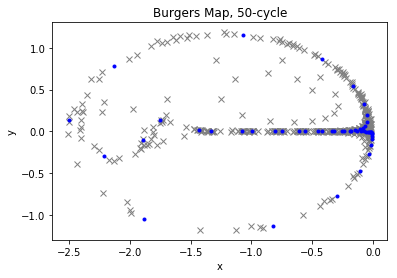}
    \end{center}
    { \hspace*{\fill} \\}

    \begin{center}
    \adjustimage{max size={0.9\linewidth}{0.9\paperheight}}{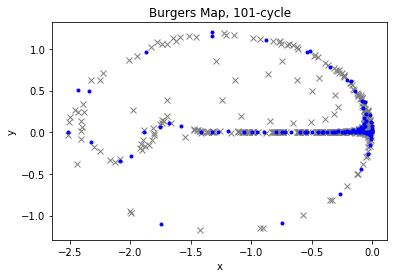}
    \end{center}
    { \hspace*{\fill} \\}

    \begin{center}
    \adjustimage{max size={0.9\linewidth}{0.9\paperheight}}{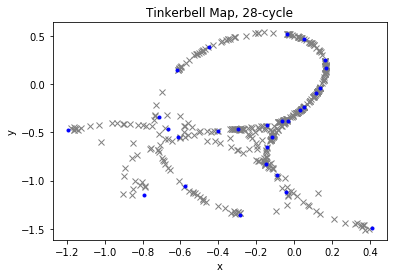}
    \end{center}
    { \hspace*{\fill} \\}

    \begin{center}
    \adjustimage{max size={0.9\linewidth}{0.9\paperheight}}{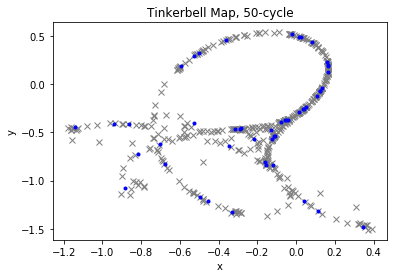}
    \end{center}
    { \hspace*{\fill} \\}

    \begin{center}
    \adjustimage{max size={0.9\linewidth}{0.9\paperheight}}{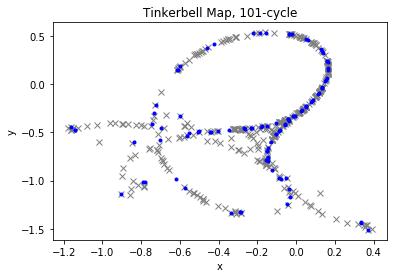}
    \end{center}
    { \hspace*{\fill} \\}

    \begin{center}
    \adjustimage{max size={0.9\linewidth}{0.9\paperheight}}{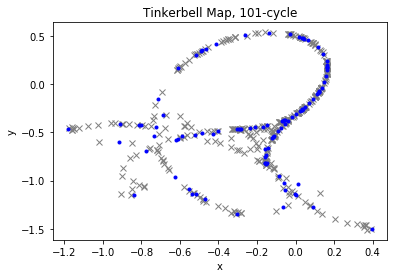}
    \end{center}
    { \hspace*{\fill} \\}

    \begin{center}
    \adjustimage{max size={0.9\linewidth}{0.9\paperheight}}{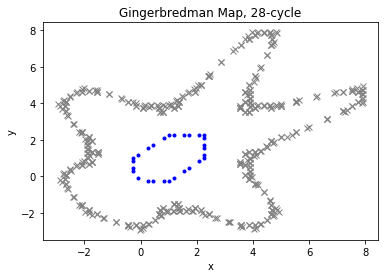}
    \end{center}
    { \hspace*{\fill} \\}

    \begin{center}
    \adjustimage{max size={0.9\linewidth}{0.9\paperheight}}{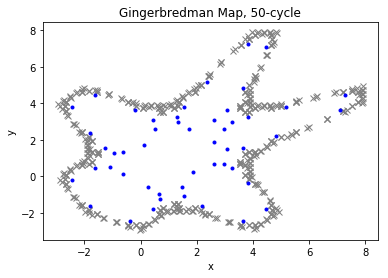}
    \end{center}
    { \hspace*{\fill} \\}

    \begin{center}
    \adjustimage{max size={0.9\linewidth}{0.9\paperheight}}{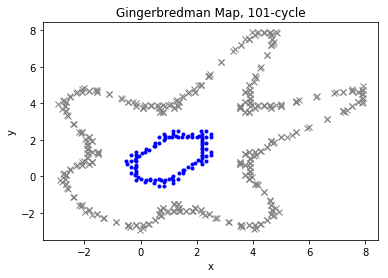}
    \end{center}
    { \hspace*{\fill} \\}

    \begin{center}
    \adjustimage{max size={0.9\linewidth}{0.9\paperheight}}{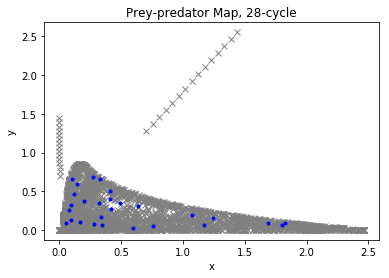}
    \end{center}
    { \hspace*{\fill} \\}

    \begin{center}
    \adjustimage{max size={0.9\linewidth}{0.9\paperheight}}{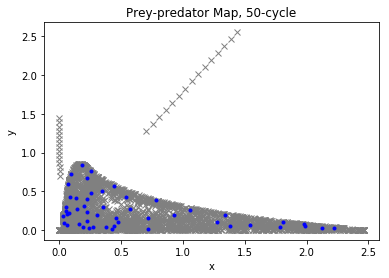}
    \end{center}
    { \hspace*{\fill} \\}

    \begin{center}
    \adjustimage{max size={0.9\linewidth}{0.9\paperheight}}{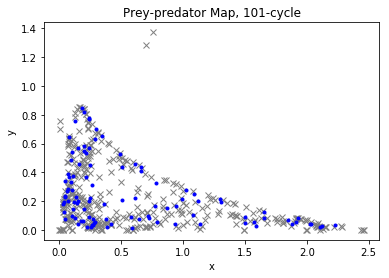}
    \end{center}
    { \hspace*{\fill} \\}

    \begin{center}
    \adjustimage{max size={0.9\linewidth}{0.9\paperheight}}{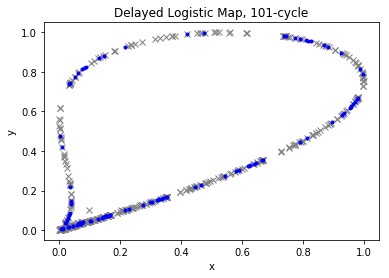}
    \end{center}
    { \hspace*{\fill} \\}

    \begin{center}
    \adjustimage{max size={0.9\linewidth}{0.9\paperheight}}{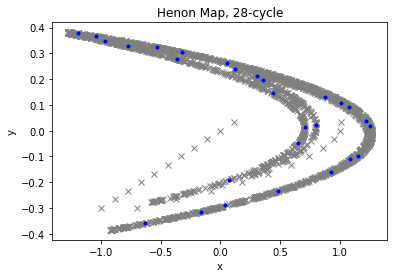}
    \end{center}
    { \hspace*{\fill} \\}

    \begin{center}
    \adjustimage{max size={0.9\linewidth}{0.9\paperheight}}{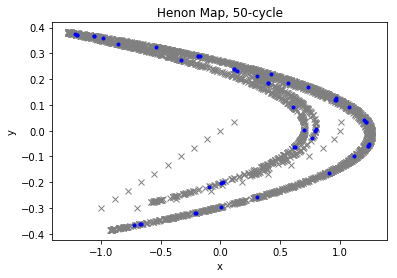}
    \end{center}
    { \hspace*{\fill} \\}

    \begin{center}
    \adjustimage{max size={0.9\linewidth}{0.9\paperheight}}{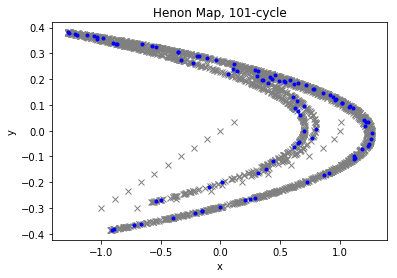}
    \end{center}
    { \hspace*{\fill} \\}

    \begin{center}
    \adjustimage{max size={0.9\linewidth}{0.9\paperheight}}{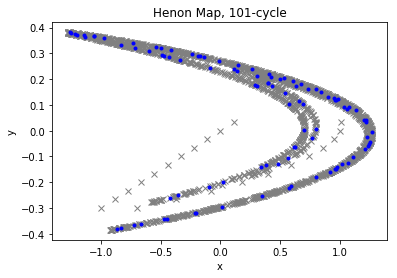}
    \end{center}
    { \hspace*{\fill} \\}

    \begin{center}
    \adjustimage{max size={0.9\linewidth}{0.9\paperheight}}{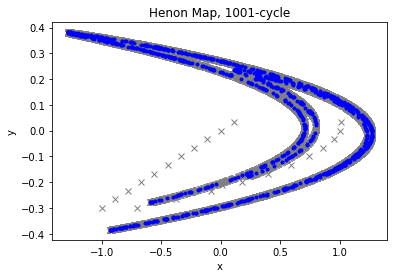}
    \end{center}
    { \hspace*{\fill} \\}

    \begin{center}
    \adjustimage{max size={0.9\linewidth}{0.9\paperheight}}{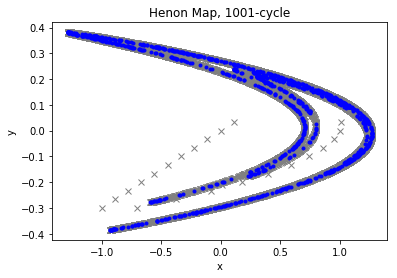}
    \end{center}
    { \hspace*{\fill} \\}

    \begin{center}
    \adjustimage{max size={0.9\linewidth}{0.9\paperheight}}{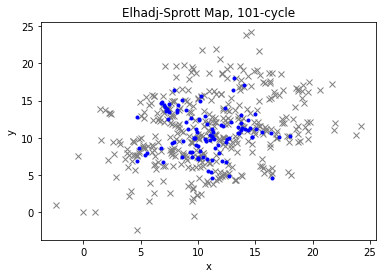}
    \end{center}
    { \hspace*{\fill} \\}

    \begin{center}
    \adjustimage{max size={0.9\linewidth}{0.9\paperheight}}{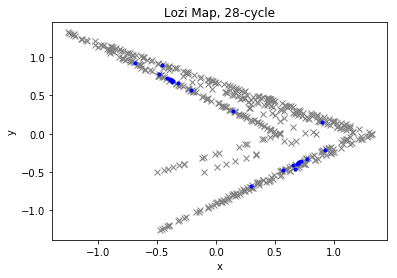}
    \end{center}
    { \hspace*{\fill} \\}

    \begin{center}
    \adjustimage{max size={0.9\linewidth}{0.9\paperheight}}{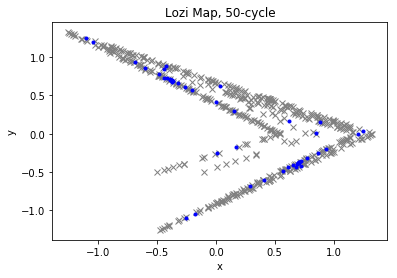}
    \end{center}
    { \hspace*{\fill} \\}

    \begin{center}
    \adjustimage{max size={0.9\linewidth}{0.9\paperheight}}{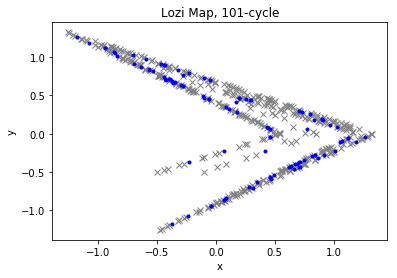}
    \end{center}
    { \hspace*{\fill} \\}

    \begin{center}
    \adjustimage{max size={0.9\linewidth}{0.9\paperheight}}{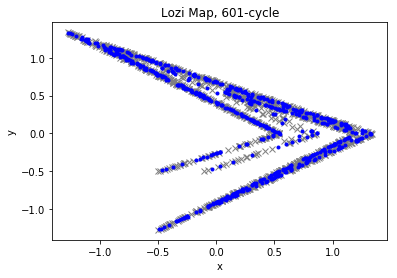}
    \end{center}
    { \hspace*{\fill} \\}

    \begin{center}
    \adjustimage{max size={0.9\linewidth}{0.9\paperheight}}{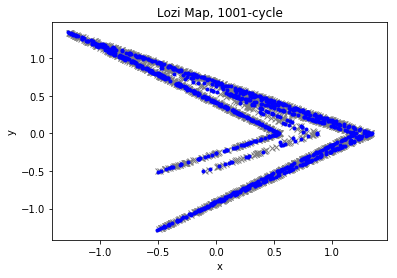}
    \end{center}
    { \hspace*{\fill} \\}

    \begin{center}
    \adjustimage{max size={0.9\linewidth}{0.9\paperheight}}{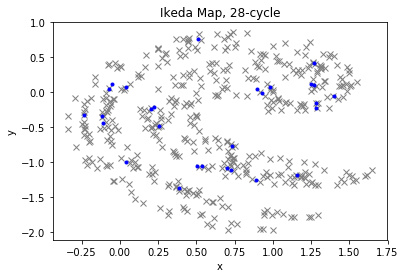}
    \end{center}
    { \hspace*{\fill} \\}

    \begin{center}
    \adjustimage{max size={0.9\linewidth}{0.9\paperheight}}{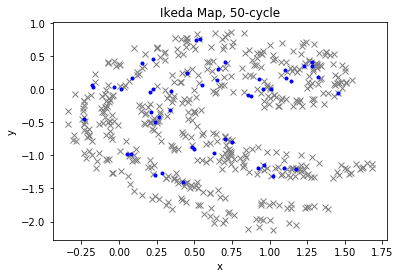}
    \end{center}
    { \hspace*{\fill} \\}

    \begin{center}
    \adjustimage{max size={0.9\linewidth}{0.9\paperheight}}{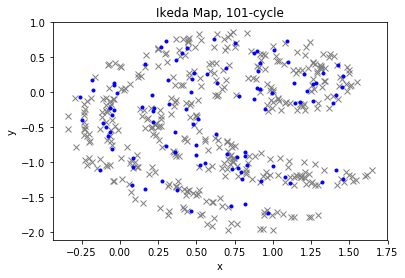}
    \end{center}
    { \hspace*{\fill} \\}

    \begin{center}
    \adjustimage{max size={0.9\linewidth}{0.9\paperheight}}{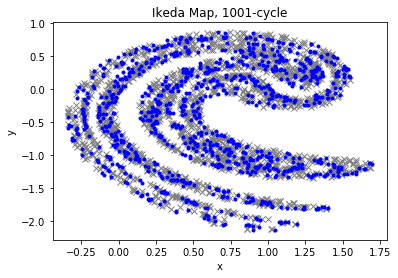}
    \end{center}
    { \hspace*{\fill} \\}

    \begin{center}
    \adjustimage{max size={0.9\linewidth}{0.9\paperheight}}{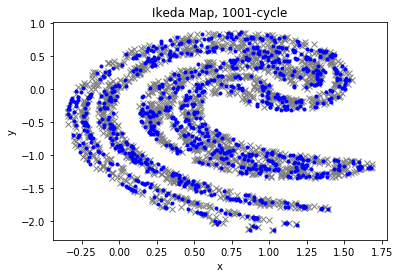}
    \end{center}
    { \hspace*{\fill} \\}

    \begin{center}
    \adjustimage{max size={0.9\linewidth}{0.9\paperheight}}{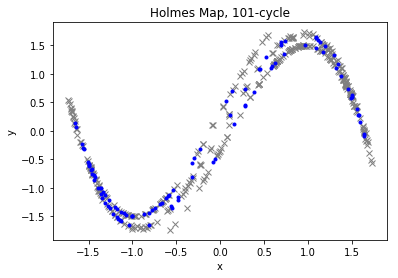}
    \end{center}
    { \hspace*{\fill} \\}

    \begin{center}
    \adjustimage{max size={0.9\linewidth}{0.9\paperheight}}{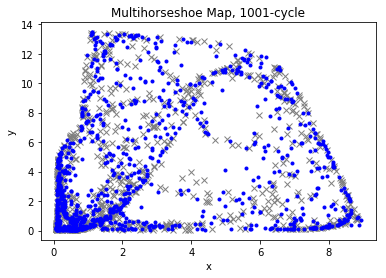}
    \end{center}
    { \hspace*{\fill} \\}

    \begin{tcolorbox}[breakable, size=fbox, boxrule=1pt, pad at break*=1mm,colback=cellbackground, colframe=cellborder]
\prompt{In}{incolor}{ }{\hspace{4pt}}
\begin{Verbatim}[commandchars=\\\{\}]

\end{Verbatim}
\end{tcolorbox}

\end{document}